\documentclass[10pt,letterpaper]{article}

\usepackage{cogsci}
\usepackage{pslatex}
\usepackage{apacite}

\usepackage{amsmath}
\newtheorem{theorem}{Theorem}
\newtheorem{corollary}{Corollary}
\newtheorem{lemma}{Lemma}

\usepackage{amssymb}

\makeatletter
\DeclareRobustCommand{\qed}{%
  \ifmmode 
  \else \leavevmode\unskip\penalty9999 \hbox{}\nobreak\hfill
  \fi
  \quad\hbox{\qedsymbol}}
\newcommand{\openbox}{\leavevmode
  \hbox to.77778em{%
  \hfil\vrule
  \vbox to.675em{\hrule width.6em\vfil\hrule}%
  \vrule\hfil}}
\newcommand{\qedsymbol}{\openbox}
\newenvironment{proof}[1][\proofname]{\par
  \normalfont
  \topsep6\p@\@plus6\p@ \trivlist
  \item[\hskip\labelsep\itshape
    #1.]\ignorespaces
}{%
  \qed\endtrivlist
}
\newcommand{\proofname}{Proof}
\makeatother
\usepackage{graphicx}

\usepackage{algorithm}
\usepackage{color}

\title{Ideal Composition of a Group for Maximal Knowledge Building in Crowdsourced Environments}
 
\author{{\large \bf Anamika Chhabra, S. R. S. Iyengar and Jaspal S. Saini}\\
 {anamika.chhabra@iitrpr.ac.in}, sudarshan@iitrpr.ac.in, jaspal.singh@iitrpr.ac.in\\
  Department of Computer Science and Engineering \\
  Indian Institute of Technology Ropar, India}

\begin{document}

\maketitle

\begin{abstract}
Crowdsourcing has revolutionized the process of knowledge building on the web. Wikipedia and StackOverflow are witness to this uprising development. However, the dynamics behind the process of crowdsourcing in the domain of knowledge building is an area relatively unexplored. It has been observed that an ecosystem exists in the collaborative knowledge building environments (KBE)\cite{chhabra2015presence}, which puts users of a KBE into various categories based on their expertise. Classical cognitive theories indicate triggering among the knowledge units to be one of the most important reasons behind accelerated knowledge building in collaborative KBEs. We use the concept of ecosystem and the triggering phenomenon to highlight the necessity for the right mix of users in a KBE. We provide a hill climbing based algorithm which gives the ideal mixture of users in a KBE, given the amount of triggering that takes place among the users of various categories. The study will help the portal designers to accordingly build suitable crowdsourced environments.

\textbf{Keywords:} 
Crowdsourcing; knowledge building; ecosystem; triggering
\end{abstract}
\section{Introduction}
Due to the introduction of Web 2.0 and Web 3.0 \cite{lassila2007embracing,hendler2009web}, a collection of tools for collaboration, integration and interaction have become available on the internet \cite{bryant2006social}. Accessing knowledge on any given topic is not a difficult task nowadays. The online knowledge building systems have become successful largely because of crowdsourcing. It is a technique provided by Web 2.0 which is used to gather information from the people in the crowd.  In a crowdsourced knowledge building system, the knowledge-generation process is outsourced to a community of users \cite{doan2011crowdsourcing}. Some of the knowledge building systems effectively exploiting the benefits of crowdsourcing are Wikipedia, StackOverflow and Quora. These systems aggregate the human knowledge on various topics and have been successful in making use of the immense potential of the masses. Many of these systems also make use of various incentivizing mechanisms to motivate users to participate more in the knowledge building process \cite{anderson2013steering}. Due to the ease of access of internet, the process of knowledge building has further evolved \cite{saxton2013rules}. Even knowledge seekers are participating in the knowledge building process by for example, asking questions \cite{wang2013wisdom, hanrahan2012modeling}. These questions then trigger other users to add their knowledge to the system \cite{dierkes2003handbook,anderson2014engaging}. Prime examples include StackOverflow and Quora. Also, group knowledge building by a mixture of experts and non-experts has mostly replaced individual knowledge building by experts \cite{kittur2008harnessing,galton1907vox,surowiecki2007wisdom}. All this has inspired us to look into the dynamics of group knowledge building. We begin by addressing the following set of questions: Why is it that  groups perform better than individuals in knowledge building? In order to have the best knowledge building experience, should we take care of certain things while forming groups? In general, does the composition of the group matter? 

\begin{figure}[!ht]
\centering
\includegraphics[scale=0.3]{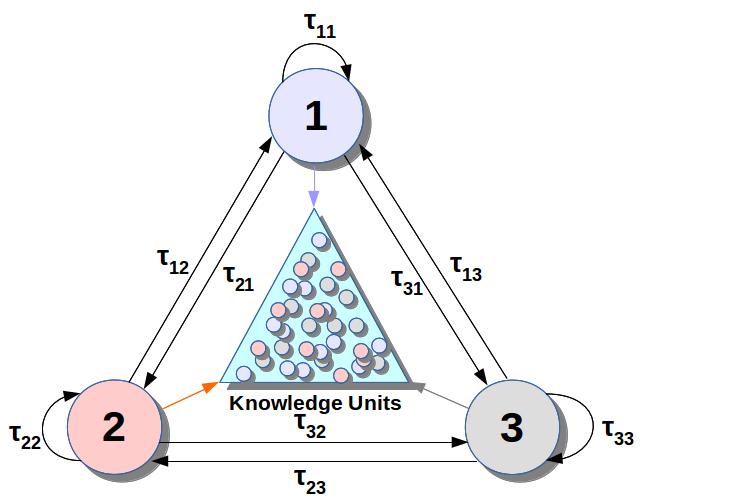}
\caption{The Triggering Phenomenon (a) A knowledge building system consisting of three categories getting triggered by each other (b) The triggering matrix $T$}
\label{tri}
\end{figure}

To investigate these questions, our analysis is based on the existence of ecosystem \cite{chhabra2015presence1} in crowdsourced KBEs, Luhmann's theory of autopoietic systems \cite{seidl2004luhmann, luhmann1995social} and Piaget's theory of equilibration \cite{piaget1977development, piaget1976piaget}. We also follow on the classical literature on the triggering phenomenon occurring among the knowledge frames. In \cite{chhabra2015presence}, the authors observed the existence of ecosystem among the users of a knowledge building system, and based on that, divided them into various categories. Each of these categories possesses a different skill set and participates in the knowledge building process in a different manner. Following Luhmann's theory, we propose that groups in a knowledge building system perform better, because these categories create an enormous amount of perturbation among each other. This perturbation leads to triggering of more ideas into the KBE. The theory when applied to modern knowledge building systems gives an insight into why people participate in these crowdsourced systems. The perturbation that is created on seeing some content on these systems inspires them to contribute more. According to Piaget's theory, this phenomenon of triggering will go on until the state of equilibration, after which the triggering reduces. This theory explains the reasons behind the slowing growth of content on Wikipedia with time \cite{suh2009singularity, lam2011past}.

The paper provides a thorough investigation of the crowdsourced knowledge building process in terms of ecosystem and triggering. We differentiate between the knowledge that the users add to the system without looking at each others' ideas and the knowledge that they add when they come across the content posted by others. We see that the latter constitutes a considerable part of the knowledge built by a crowdsourced system. The proposed model helps in understanding the reason behind the benefits that a group provides, as compared to individual knowledge building. We further show that it is important to have the right mix of users in a knowledge building portal for it to be more effective. For example, even the users who have good knowledge of the subject matter, might not add their knowledge to the system unless they are instigated. Hence, it is equally important to have a good number of users in the system who can ask meaningful questions. This way, even people having less knowledge about the subject participate in the knowledge building process. 

{\textit{Outline.} The remainder of the paper is organized as follows: In section 2, we discuss related work on knowledge building. Section 3 first explains certain classical theories which have inspired the development of the current model along with the model and then gives an algorithm to find the ideal distribution of users given the triggering matrix. In section 4, we conclude and discuss directions for future work.}

\section{Related Work}
In the past, most of the models on knowledge building have focused on various characteristics that a KBE should possess. However, the characteristics of users that are participating in the knowledge building process have not been well understood.

Nonaka \citeyear{nonaka1994dynamic} explains that the knowledge creation happens through a constant dialogue between tacit knowledge\footnote{It is codified knowledge which is transmittable in some formal language} and explicit knowledge\footnote{It is the knowledge that human beings possess and acquire through their experience over time}. The author developed a theoretical framework which provides an analytical perspective on various dimensions of knowledge creation. Gerry Stahl \citeyear {stahl2000model} provides a conceptual framework for collaborative KBEs which consists of important phases that should be supported by any computer-based KBE. The model given by the author explains the relationship of collaborative group processes to individual cognitive processes. Scardamlia and Brieter \citeyear{scardamalia1994computer} in their work consider learners also as the members of knowledge building community. They emphasized the importance of a community effort rather than individual effort for accelerated knowledge building. Cress et al. \citeyear{cress2008systemic} provide a theoretical framework for describing the process of learning and knowledge building and the way these two processes influence each other. They make use of Luhmann theory and Piaget's theory of equilibration and propose two processes externalization\footnote{Externalization is a process by which users add their knowledge to the system} and internalization\footnote{Internalization is the process of taking the information from the system} as the basis of interaction between the social and the cognitive system. Minsky \citeyear{minsky1977frame} and Rumelhart \citeyear{rumelhart1991understanding} consider the knowledge units as frames. Norman \citeyear{norman1981categorization} and Just et. al \citeyear{just1980theory} assert that the frames are linked together and one frame may trigger other frames. When a frame is triggered, all the frames linked to it are triggered. Fisher and Lipson \citeyear{fisher1985information} state that for a given frame to come into the system, suitable triggering conditions are required.

A study which is conducted in the context of a right mix of users in the domain of problem solving has been performed by Scott \citeyear{hong2004groups, page2008difference}. The author states that a group of randomly selected people outperforms a group of best performing people. This is due to the fact that these random people bring diverse knowledge into the system and hence are able to perform better while solving a problem. On contrary, the best performing agents bring in similar type of knowledge to the system and hence might not be able to solve the problem that well. However, Thompson \citeyear{thompson2014does} came up with a counter paper to Scott's work recently claiming that his paper does not provide any foundation for the argument that diversity actually trumps ability. Krause et al. \citeyear{krause2011swarm} argue that adding diversity to a group can be more advantageous than adding expertise to the group. Erickson et. al \citeyear{erickson2012hanging} provide a framework to select the crowd matching organizational needs. The authors state that different tasks require different crowds with different skills and knowledge. Kobern et. al \cite{kobren2015getting} observe in their recent work that intelligently assigning tasks to the users significantly increases the value of a crowdsourced system.

\section{The Triggering Model} 
This section first introduces the motivation behind the model by presenting classical theories which explain the phenomenon of triggering among the knowledge units and a brief of the existence of ecosystem in crowdsourced environments. It then describes the types of KUs followed by the problem statement. The Triggering Model and the algorithm are explained afterwards.
\subsection{Motivation}
This subsection introduces Luhmann theory, Piaget's theory of equilibration and the previous work on ecosystem.
\subsubsection{Elementary Theories on Cognition and Triggering}
Luhmann’s theory describes a knowledge building system as an autopoietic system \footnote{ Autopoietic system refers to a system which once started, keeps recreating and maintaining itself.}. It states that once started, further ideas (or cognitions) are produced  by existing ideas of the same system. The existing ideas create perturbations in the cognitive system of users, which then trigger more ideas.The theory clearly distinguishes cognitive systems from the KBE. We apply the same concept to the knowledge building process, where users are the cognitive systems who build knowledge in the KBE. The changes in the environment lead to irritations in the cognitive systems. The theory further talks about the concept of \textit{structural coupling}, which is the relation between cognitive systems and the environment. It states that although environment can create irritations in all the cognitive systems (users here), it might not be able to trigger all of them. The actual systems that can get triggered due to these perturbations are determined by the structural coupling of these systems to the environment. Also, different cognitive systems may have different structural coupling.

The question arises whether this phenomenon of triggering goes on indefinitely, or does it reach some threshold point. Piaget’s Model of equilibration states that people contribute to the knowledge building process because of cognitive conflicts (perturbations as per Luhmann Model), which means that when they see some information that is incongruent to their existing knowledge, it creates a disturbance in their mind. This disturbance leads them to add their knowledge to the system.  This knowledge addition then leads to the equilibration between the system’s knowledge and the user’s knowledge. We further build on this theory explaining that initially, there is a huge need for equilibration, hence enormous amount of knowledge gets added to the system. However, as time passes, users’ knowledge starts matching with the knowledge of the system, and hence, there is less disturbance. This provides some support for slowing growth of knowledge building in the KBEs with time.  

Lateral thinking and triggering process \cite{rumelhart1991understanding} are the main reasons for generation of ideas in a collaborative KBE. Triggering is a procedure by which an idea or a comment spearheads the generation of another idea or a thought \cite{just1980theory}. It is associated with how an idea or a comment becomes an impetus for the generation of other ideas. The generation of knowledge units in terms of triggering has also been explained in the information processing theory \cite{norman1981categorization, seidl2004luhmann}.  This theory characterizes human knowledge as a series of ‘Knowledge Frames’. The frames are related to each other by various conditions. When a frame is triggered by stimuli, other frames, which are linked to this frame, may also be triggered. The frames play an important role in guiding the way for creating and retrieving more knowledge.

\subsubsection{Ecosystem and Formation of Categories}
In \cite{chhabra2015presence}, the authors conducted experiments on a custom-made annotation system CAS and observed that people showed expertise in performing mainly one kind of activity while participating in the knowledge building process. The existence of ecosystem was observed in Wikipedia and StackOverflow as well. This phenomenon divides all the users into various categories based on their expertise. For example, the categories observed in CAS (Crowdsourced Annotation System) were Articulators, Probers, Solvers and Explorers. There were different categories observed in the case of Wikipedia and StackOverflow. These observations led to the conclusion that ecosystem is a characteristic of crowdsourced KBEs. The categories formed in the ecosystem trigger each other with different triggering factors. Figure~\ref{tri} shows the triggering phenomenon occurring in a three category system. $\tau_{ij}$ is the triggering factor, which quantizes the number of knowledge units of category $i$ that will be generated per user of category $i$ due to one knowledge unit of category $j$. This gives rise to a triggering matrix storing the triggering factors for all the categories. The triggering matrix for a three category scenario is shown in fig. 1(b). We intuitively believe that some values in this matrix will be low, high or some will even be zero, depending on whether one category is a prerequisite for the generation of KUs of another category or not. For example, in CAS, a question triggers users in `Solver' category more than it may trigger the users in `Pointer' category. We call triggering within categories \textit{`Intra-triggering'} and across categories as \textit{`Inter-triggering'}. The values of intra-triggering are perceived to be less than those of inter-triggering due to the similarity of traits among the users of same category. The self loops of categories in the diagram represent intra-triggering.

A collaborative knowledge building system, due to the presence of ecosystem, provides an ideal environment for the triggering process to take place. The users in one category trigger the users in all other categories. Due to this triggering process only, these collaborative environments are able to realize `the whole is greater than sum of its parts' phenomenon.

\subsection{Sources of Knowledge Units in a Knowledge Building System}
We can classify all the knowledge generated in the system on the basis of whether it is an outcome of group dynamics or not. The total generated knowledge can be divided into two categories: (i) Internal Knowledge and (ii) Triggered Knowledge. 

\subsubsection{Internal Knowledge}
Internal knowledge is a subset of the user's knowledge which is added to the system independent of the effect of group dynamics. This is precisely the knowledge that the user would have added to the system if she had been participating in the knowledge building process individually (and not in a group). Hence, addition of internal knowledge to the system consists of only the process of externalization. As an example, consider the following experiment: If a user is asked to name all the countries in the world (which are more than 190 in total), assume, she is able to come up with 40-50 of these countries. These generated knowledge units are what we consider as her internal knowledge. Please note that in this case, it may so happen that the user knows some more countries' names, but currently she does not recall them. These names are not a part of internal knowledge since they never got added to the system. 

Please note that different users may have overlapping internal knowledge. This hints us towards the fact that there exits an upper bound on the internal knowledge contribution of a particular category, i.e. after a certain threshold number of users of category $i$, adding more users to this category will not lead to more internal knowledge contribution of the category. This shows that the `law of diminishing benefits' holds here rather than the `law of large numbers'. This information may help us to find an upper bound on the maximum number of users in a particular category.

\subsubsection{Triggered Knowledge}
This is the kind of knowledge that gets added to the system as a result of the group dynamics. When people participate in the knowledge building process as a group, they get triggered on seeing each others' ideas and hence, generate more knowledge. This knowledge is called triggered knowledge. Addition of triggered knowledge to the system is a process, in which first internalization and then externalization takes place. The users first internalize the knowledge from the system, and then externalize their own knowledge to thee system. This externalized knowledge can further be of two types:
\begin{enumerate}
\item First type consists of the knowledge that the user had in her cognitive system already, but she gets reminded of it only after getting triggered by some other user's idea. In our previous example on countries' names, for example, it may so happen that on hearing about some country of Europe, she gets reminded of some other countries' names from Europe. This type of knowledge is called \textit{`Recollected Knowledge'}.
\item The other type consists of the new knowledge which was not in users' cognitive system already. This type of knowledge has also been referred to as \textit{Emergent Knowledge} \cite{cress2008systemic}. For example, consider the scenario where a user knows these two facts: (a) There are infinite Fermat Numbers. (b) Every number can be uniquely decomposed into a product of prime numbers (Fundamental Theorem of Arithmetic). Now consider that she comes across another fact which was added to the system by some other user, which is, (c) Every two Fermat Numbers are co-prime to each other. The user combines her existing knowledge ((a) and (b)) and the knowledge internalized from the system ((c)) and comes up with a new knowledge, for example, that there are infinite prime numbers. This new deduced knowledge when added to the system, also falls under the category of triggered knowledge.
\end{enumerate}

\subsection{Problem Statement}
\textbf{Problem:} \textit{Given the number of users, the number of categories, the total internal knowledge of the users and the amount of triggering taking place among the categories, can we output an ideal number of users for each category that leads to the maximum knowledge building and hence accelerates the knowledge building process?}

To motivate the discussion of the problem, let us assume a simplistic two-category scenario as shown in the figure~\ref{two_cat}. As we can see that when the users of category $B$ add a KU to the system, the users of category $A$ are triggered to a very less extent (i.e. $0.1$). So does this mean that we should keep as less as possible number of users in category $A$ However, we also see that whenever users of category $A$ add a KU to the system, it triggers the users of category $B$ to a great extent (i.e. $0.9$). Now, if it is given that we can take only $100$ total users in the system, how many users should we assign to both the categories so that the total knowledge generated is maximum. Surprisingly, in this case, 50 users in each of the categories will lead to the maximum knowledge building if it is given that all the users have same amount of internal knowledge. In general if we are given the amount by which various categories trigger each other, can we find out an ideal number of users that should be assigned to each category for maximum knowledge building?

\begin{figure}[!ht]
\centering
\includegraphics[scale=0.3]{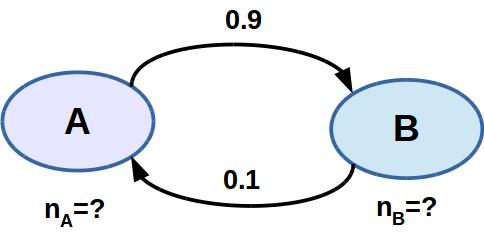}
\caption{A Simplistic two-category Case}
\label{two_cat}
\end{figure}

To get an answer to the problem explained above, we will first build a model of crowdsourced knowledge building which makes use of the triggering phenomenon, ecosystem and the two types of knowledge units (internal and triggered). With the help of this model, we will show that different distributions of users across categories gives rise to varying amount of total knowledge produced. We will then develop an algorithm that takes the triggering matrix as input and gives the ideal distribution of users leading to the maximal knowledge building.

\subsection{The Model}
Consider a knowledge building system (KBS) consisting of $n$ users. We assume that ecosystem exists in the KBS i.e. each users tends to be an expert \textit{mainly} in one particular type of activity. Further we consider a stricter version of the above scenario where each user belongs to \textit{only} one category i.e. a user will generate only one type of knowledge units, although it can trigger users to generate varied types of knowledge units. This assumption has been made to reduce the complexity of the model. In practice, it may happen that the users contribute to rest of the categories to a small extent. Given that the current model mainly concentrates on the triggering phenomenon in KBS, incorporation of this small amount of contribution of knowledge to the other categories will make the model extremely complex, and may be considered in the next version of the model. We assume that all the $n$ users enter the knowledge building system at time $t=0$. For example, in settings like NB \cite{zyto2012successful} and CAS \cite{chhabra2015presence}, all the users are usually present in the system at the start of the knowledge building process. Throughout the text, we will be using the words `type' and `category' interchangeably.

Let $n_i$ be the number of users in category $i$ and $m$ be the number of categories such that,
$$\sum_{i=1}^{m}n_i=n$$
Let $r_i$ be the internal knowledge contribution of a user of category $i$. For simplicity, we assume that all the users of a category add same amount of internal knowledge to the system. This assumption has been taken so that the model can concentrate on explaining the most important feature which is triggering among the categories. This assumption, however, can be easily relaxed without significant changes to the model. The users of a category will add their internal knowledge to the system in discrete time stamps depending on the type of knowledge building system under consideration. In some cases, they may add it in the beginning. In other cases, they may keep adding it gradually to the system as some function of time. Let $r_i(t)$ be the internal knowledge contribution of a user of category $i$ at time $t$ such that,

$$\sum_{t=0}^{\infty}r_i(t)=r_i$$

Let $\tau_{ij}$ be the number of knowledge units (KUs) of type $i$ that get triggered due to one KU of type $j$. Now, our aim is to calculate the total number of KUs of type $i$ that get added to the KBS at time $t$. Let $k_i(t)$ represent that number. We consider that the number of KUs of a category that get added at time $t$ are directly dependent on the following parameters:
\begin{enumerate}
\item The number of KUs of all the categories\footnote{Including the category under consideration, due to some amount of intra-triggering} that get added to the system at time $t-1$. This is because of the fact that all the categories' KUs will trigger new KUs of the considered category at the next time step.
\item The triggering factors from all other categories to the considered category.
\item The number of users in the considered category. More the number of users, more the KUs generated\footnote{Upto a certain threshold, after which diminishing benefits appear due to overlap of knowledge of users of the same category}.
\item The internal knowledge of the users of the considered category.
\end{enumerate}

These parameters give rise to the following expression for $k_i(t)$,
$$k_i(t)=n_i(\tau_{i1}k_1(t-1)+\tau_{i2}k_2(t-1) + \dots + \tau_{im}k_m(t-1))+n_ir_i(t)$$
i.e.,
\begin{equation}\label{eq:1}
k_i(t)=n_i\left(\sum_{j=1}^{m}\tau_{ij}k_j(t-1) + r_i(t)\right) 
\end{equation}
$\forall 0 \leqslant i \leqslant m$

The above expression gives us the number of KUs of category $i$ that get added to the system. We are now interested to similarly compute the KUs of all the categories that get added to the system at time $t$.\\
Let $K(t)$ represent the column vector consisting of the knowledge generated by various categories at time $t$ as its elements.
\[ 
k(t)=
\begin{bmatrix}
k_1(t)\\
k_2(t)\\
\vdots\\
k_m(t)
\end{bmatrix}
\] 

Let $T$ be the triggering matrix, $N$ be a diagonal matrix storing the number of users in each category and $R(t)$ be the column matrix storing the function by which each category users add their internal knowledge to the system. The matrices $T$, $N$ and $R(t)$ are shown as below:
\[
T=
\begin{bmatrix}
\tau_{11} & \tau_{12}&\cdots &\tau_{1m}\\
\tau_{21} & \tau_{22}&\cdots&\tau_{2m}\\
\vdots & & &\vdots\\
\tau_{m1}& \cdots & \cdots & \tau_{mm}\\
\end{bmatrix}
\]

\[N=
\left(
 \begin{array}{ccccc}
   n_1\\
    & n_2 & & \text{\huge0}\\
    & & \ddots\\
    & \text{\huge0} & & 1\\
    & & & & n_m
 \end{array}
\right)
\]
 
\[
R(t)=
\begin{bmatrix}
r_1(t)\\
r_2(t)\\
\vdots\\
r_m(t)\\
\end{bmatrix}
\]
 
Given the above notations, we can write equation~\ref{eq:1} as:
\begin{equation}\label{eq:2}
K(t)=N(TK(t-1) + R(t))
\end{equation}
where $K(t)$, $N$, $T$ and $R(t)$ are the matrices as defined above. The above equation gives us a recursive formula for finding the total number of KUs added to the system at time $t$. Our aim now is to calculate the total KUs of each category added to the system, when $N$, $T$ and $R$ are given, where $R$ is defined as below:
\begin{align} \label{mat:R}
R=
\begin{bmatrix}
r_1\\
r_2\\
\vdots \\
r_m\\
\end{bmatrix}
\end{align}
where $r_i$s represent total internal knowledge of category $i$. We will soon show that the total knowledge generated in the system is independent of the function $r(t)$ for any category, i.e. in whichever way the users may add their internal knowledge to the system, the total knowledge generated in the system over a sufficiently long period of time in all cases comes out to be equal.

\begin{lemma}
The total knowledge generated in the system at time $t$ is given by:
$$ K(t)=\sum_{i=0}^{t}(NT)^iNR(t-i)$$
\end{lemma}
\begin{proof}
Substituting the value of $K(t-1)$ in equation~\ref{eq:2}, we get,
\begin{align*}
K(t)&=N(T(NTK(t-2) + NR(t-1)) + R(t))\\
K(t)&=(NT)^2K(t-2) + (NT)NR(t-1) + NR(t)\\
\end{align*}
Continuing like this, we get,
\begin{align*}
K(t)&=(NT)^tK(0) + (NT)^{t-1}NR(1) + \\&(NT)^{t-2}NR(2) + \dots + (NT)^0NR(t)\\
\end{align*}
Since at time $t=0$, no triggering happens, and the total knowledge of the system is given by the internal knowledge of all the users, we take $K(0)=NR(0)$, which gives us,
\begin{align*}
K(t)&=(NT)^tNR(0) + \dots + (NT)^0NR(t)\\
K(t)&=\sum_{i=0}^{t}(NT)îNR(t-i)\\
\end{align*}
\end{proof}
In order to get the total KUs ever generated in the KBS, we consider another matrix $K$ as follows:
\[K=
\begin{bmatrix}
k_1\\
k_2\\
\vdots\\
k_m
\end{bmatrix}
\]
where $k_i$s are scalar values storing the total number of KUs of category $t$ ever generated in the KBS.
\begin{theorem}\label{th:1}
Given the matrices $N$, $T$ and $R$, the net knowledge in the system at the end of the Knowledge building process is given by,
$$K=(I-NT)^{-1}NR$$ 
assuming $$\rho(NT)< 1$$ where $I$ is an Identity Matrix of the order $m \times m$.
\end{theorem}
\begin{proof}
Lemma\ref{eq:2} gives the KUs generated at time $t$ for all the categories. Summation of $K(t)$ over the time period from $0$ to $\infty$ gives,
\begin{align*}
K&= \sum_{t=0}^{\infty} \left(\sum_{i=0}^{t}(NT)^{i}NR(t-i)\right) \\
K&=\sum_{i=0}^{t}\sum_{t=0}^{\infty}(NT)^{i}NR(t-i) \\
K&= \sum_{i=0}^{t}(NT)^{i}NR \text{, where $R$ is as defined in (\ref{mat:R})}\\
\end{align*}
If the spectral radius of $(NT)$ is less than $1$, $\sum_{i=0}^{t}(NT)^{i}$ converges to $(I-NT)^{-1}$, i.e.
\begin{align*}
K&=(I-NT)^{-1}NR \text{, } 
\end{align*}
where $I$ is an Identity matrix.
\end{proof}

In the above proof, $\rho(NT)$ has been taken to be less than 1. This is because we are assuming a bounded system. If we take $\rho$ to be greater than or equal to 1, then the term $\sum_{i=0}^{t}(NT)^{i}$ becomes increases exponentially, which is not a practical scenario. 

The above theorem shows that the total knowledge in the system is independent of the distribution of $R(t)$ i.e. whether the entire internal knowledge of users enters the system at $t=0$ or it enters in discrete time, the net knowledge in the system remains unchanged. The theorem also helps computing the triggered knowledge contribution in the knowledge building, as shown by the corollary that follows.
\begin{corollary}
The contribution of triggered knowledge to the total knowledge of the system is given by $N.T.K$
\end{corollary}
\begin{proof}
From Theorem \ref{th:1}, 
\begin{align*}
K&=(I-NT)^{-1}NR\\
\implies (I-NT)K&=NR\\
\end{align*}
The contribution of triggered knowledge in $K$ is equal to $K-NR$, where $NR$ is the total internal knowledge added to the system.
\begin{align*}
\implies K-NTK&=NR\\
\implies K-NR&=NTK\\
\end{align*}
\end{proof}
The corollary indicates that the ratio of triggered knowledge to the total knowledge is directly proportional to the triggering matrix $T$. This means that the total knowledge that a system is able to produce depends on how much the users of the system are able to trigger each other.

In order to know the shape of the surface curve when we change the number of users across the categories, we simulated the results of Theorem ~\ref{th:1} by taking the following parameters:

$n=100$\\
$m=3$\\
$
T=
\begin{bmatrix}
0.001 & 0.003 & 0.0004\\
0.03 & 0.001 & 0.005 \\
0.03& 0.008 & 0.003\\
\end{bmatrix}
$

For simplicity, the internal knowledge ($R$) for all the categories was taken to be equal. It was observed that the surface plot comes out to be a convex plot which exhibits a peak corresponding to a particular distribution of users. Figure~\ref{surface} shows the plot where the axes for the represent $n_1$, $n_2$ and $K$ ($n_3$ is taken as $(n-n_1-n_2)$). The maxima of the curve was found at $K = 3081.40849516$ with the distribution $[39,30,31]$.
\begin{figure}
\centering
\includegraphics[scale=0.50]{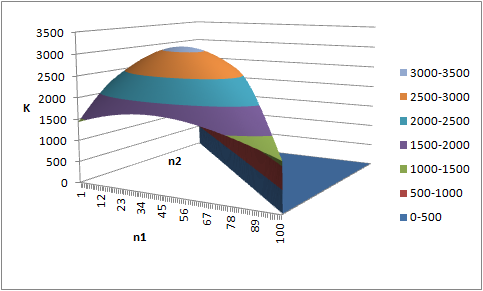}
\caption{The convex plot of net knowledge (K) generated across all the possible distributions}
\label{surface}
\end{figure}

\subsection{Hill Climbing Algorithm For Ideal Skillset Distribution}
Since we obtained a convex surface plot with respect to the distributions, we can develop a hill climbing algorithm which starts from any random distribution and gradually moves towards the ideal distribution of users which leads to an accelerated knowledge building. We now provide a `Hill Climbing Algorithm' (See Algorithm 1) for finding the ideal skill set distribution across categories. 

\begin{algorithm}
\caption{Hill Climbing Algorithm For Ideal Distribution of Users across Categories}\label{algo1}

%

\textbf{Input:} Number of users $n$, Number of categories $m$, Triggering Matrix $T$ and Internal Knowledge matrix $R$. \\
\textbf{Output:} Ideal Distribution Matrix $N$
\begin{enumerate}
\item Pick a random distribution $N=[n_1,n_2,\dots,n_m]$ of users across different categories, such that $\sum_{i=1}^{m}n_i=n$
\item Consider all the $2*{m\choose{2}}=m(m-1)$ neighbouring distributions ($Nb$) of D, defined as:

$Nb=[~]$\\
$for~ i=1:m$\\
$~~~~~~for~ j=1:m$\\
$~~~~~~~~~~If ~i\neq j:$\\
$~~~~~~~~~~~~~~D=(n_1, n_2, \dots , n_i-1, \dots , n_j+1, \dots , n_m)$\\
$~~~~~~~~~~~~~~Nb.append(D)$

\item Calculate $K_{Nb[j]}~ \forall~ 1\leqslant j \leqslant m(m-1)$	based on the following formula:
$$K=(I-NT)^{-1}NR$$	\\
where $K_{Nb[j]}$ represents the net knowledge in the system with the distribution $Nb[j]$

Choose $i$ $such~ that~ K_{Nb[i]}=\displaystyle{\max_j}~K_{Nb[j]}$
\item If $K_{Nb[i]}>K_N$ then $N=Nb[i]$ and we repeat steps $2$ to $4$\\
else:\\
return $N$ //The ideal skill set distribution
\end{enumerate}
\end{algorithm}

\begin{figure}
\centering
\includegraphics[scale=0.4]{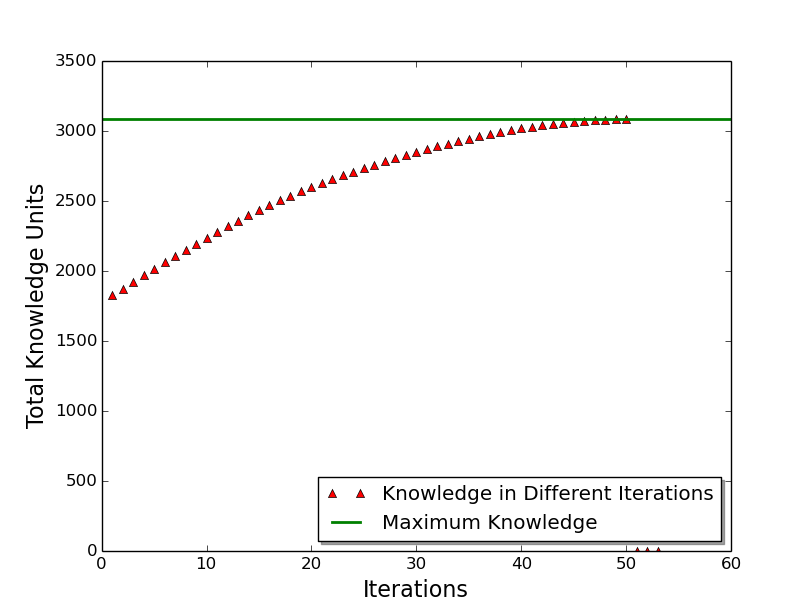}
\caption{The execution of Hill Climbing Algorithm: moving from one distribution to another neighboring distribution, leading to the perfect one}
\label{hill}
\end{figure}

For verification, we simulated the proposed algorithm on the same specifications as defined above. Figure~\ref{hill} shows how the algorithm moves from an initial distribution (which is input to the algorithm) to ideal  distribution i.e. the distribution which produces maximum knowledge in the given KBE. The final output of the algorithm ($[39,30,31]$) is coherent with results obtained from the previous simulation, supporting the correctness of the proposed algorithm.

\section{Conclusion and Future Work}
This paper is a step towards understanding the dynamics of knowledge building in a crowdsourced environment. We consider the variance of expertise present in the crowd and divide them into categories. In line with the previous research in cognitive sciences, we consider triggering among the knowledge frames to be the reason behind better performance of the groups than individuals. Through a model, we explain the interplay of these categories taking into account the varied amount of triggering that takes place among them. We then compare this triggered knowledge with the knowledge that would have been added to the system in the absence of interaction among the users. We find that triggered knowledge is a significant part of the total knowledge produced.

We emphasize the importance of a right mix of skill set present across the categories in order to accelerate the knowledge building process. We observe that as we keep changing the distribution of users across categories, we get different amount of knowledge produced with respect to time. In that context, we try to find the ideal distribution which accelerates the knowledge building process to the best possible extent. We computationally observe that the knowledge curve with respect to the distribution of users is a convex curve, the peak of which gives the ideal distribution. We further develop a Hill Climbing Algorithm which starts from any random distribution and outputs the ideal distribution for the best knowledge building experience. We believe that this study will inform the portal designers to take better decisions while developing a KBE. For example, knowing the number of categories in the system and the triggering among them, they can use various incentivizing mechanisms to encourage the participation of users in different categories. This theory can also help explain why some KBEs like Wikipedia have flourished while some others like could not proliferate.

In future, the proposed model may be verified for various KBEs like Wikipedia, Quora, StackOverflow etc. The model could also be tested by changing the existing distribution of users in a KBE, and observing the impact of this change towards the acceleration/deceleration in the knowledge building process. In the model, all the users were assumed to be present in the system at the start of the knowledge building process, which may not always be the case. Therefore the model may be further improvised to account for this parameter.

\bibliographystyle{apacite}

\setlength{\bibleftmargin}{.125in}
\setlength{\bibindent}{-\bibleftmargin}

\bibliography{ecosystem_cogsci}

\end{document}